\documentclass[peerreview,12pt]{IEEEtran}

\usepackage{ifpdf}
\ifpdf
	\usepackage[pdftex]{graphicx}
  \usepackage[update]{epstopdf}
\else
	\usepackage{graphicx}
\fi
\usepackage{amsmath,amssymb,array,stfloats,cuted,midfloat,psfrag,setspace,mathbbol}
\usepackage[noadjust]{cite}

\newtheorem{theorem}{\textbf{Theorem}}
\newtheorem{prop}[theorem]{\textbf{Proposition}}

\newtheorem{definition}{\textbf{Definition}}

\newtheorem{proof}{\textbf{Proof}}
\newcommand{\Expt}{\mbox{${\mathbb E}$} }

\begin{document}
\date{}

\title{The Achievable Distortion Region of Bivariate Gaussian Source on Gaussian Broadcast Channel}

\author{Chao Tian,~\IEEEmembership{Member,~IEEE}, Suhas Diggavi,~\IEEEmembership{Member,~IEEE}, \\and Shlomo Shamai (Shitz),~\IEEEmembership{Fellow,~IEEE}
}

\maketitle

\begin{abstract}
We provide a complete characterization of the achievable distortion region for the problem of sending a bivariate Gaussian source over bandwidth-matched Gaussian broadcast channels, where each receiver is interested in only one component of the source. This setting naturally generalizes the simple single Gaussian source bandwidth-matched broadcast problem for which the uncoded scheme is known to be optimal. We show that a hybrid scheme can achieve the optimum for the bivariate case, but neither an uncoded scheme alone nor a separation-based scheme alone is sufficient.  We further show that in this joint source channel coding setting, the Gaussian setting is the worst scenario among the sources and channel noises with the same covariances. 
\end{abstract}

\section{Introduction}
\label{sec:intro}
It is well known that Shannon's source-channel separation result for point-to-point communication \cite{Shannon:48} does not hold for in general for multi-terminal systems, and thus joint source-channel coding may be required to achieve the optimum. One simple yet intriguing scenario where source-channel separation is known to be suboptimal is broadcasting Gaussian sources on Gaussian channels. 

When a single Gaussian source is at the encoder, the achievable distortion region is known when the source bandwidth and the channel bandwidth are matched \cite{Goblick:65}, for which a simple analog scheme is optimal. However when the source and channel bandwidths are not matched, exact characterization of the achievable distortion region is not known. The best known coding schemes are based on joint source-channel using hybrid signaling \cite{ReznicFederZamir:06,Prabhakaran:05}, and approximate characterizations were given in \cite{TianDiggaviShamai:09}; see references therein for related works. As a simple extension to this problem of single Gaussian source broadcasting, when the source is a bandwidth-matched bivariate Gaussian and each decoder is interested in  one source component, only partial characterization is known when uncoded scheme is shown to be optimal under certain signal-to-noise ratio conditions \cite{Bross:08}. 

In this work, we provide a complete characterization of the achievable distortion region for broadcasting bivariate Gaussian sources over Gaussian channels when each receiver is interested in only one component, where the source bandwidth and the channel bandwidth are matched. We further show that in this joint source channel coding setting, the Gaussian problem is the worst scenario among the sources and channel noises with the same covariances, in the sense that any distortion pair that is achievable in the Gaussian setting is also achievable for other sources and channel noises. Our work is built on the outer bounds given in \cite{Bross:08} and we show that a hybrid coding scheme (different from the one given in \cite{Vishwanath:09} proposed for the same problem) can achieve the outer bounds. Our main contribution in this work is this new coding scheme and a detailed and systematic analysis of the inner and outer bounds, which result in a complete solution. To the best of our knowledge, this is the first case in the literature that a hybrid scheme is shown to be optimal for a joint source-channel problem, whereas neither an uncoded scheme alone nor a separation-based scheme alone is optimal. 

\section{Problem Definition}
\label{sec:definition}

Let $\{S_1(i),S_2(i)\}$ be a memoryless and stationary bivariate Gaussian source with zero mean and covariance matrix
\begin{align}
\label{eqn:covariance}
\bigg{(}
\begin{array}{cc}
\sigma^2 &\rho \sigma^2\\
\rho\sigma^2 & \sigma^2\\
\end{array}
\bigg{)}
\end{align}
where without loss of generality, we can assume $\rho\geq 0$. The
vector $(S_k(1),S_k(2),...,S_k(n))$ will be written as $S^n_k$ for $k=1,2$. We use $\mathbb{R}$ to denote the domain of reals. The Gaussian memoryless broadcast channel is given by the model 
\begin{align}
Y_k=X+Z_k,\quad k=1,2,
\end{align}
where $Y_k$ is the channel output observed by the $k$-th receiver, and $Z_k$ is the zero-mean independent additive Gaussian noise in the channel. Thus the channel is memoryless in the sense that $(Z_1(i),Z_2(i))$ is a memoryless and stationary process. The variance of $Z_k$ is denoted as $N_k$, and without loss of generality, we shall assume $N_1\leq N_2$. Throughout the paper, we use natural logarithm.

\begin{figure}[tb]
\begin{centering}
\includegraphics[width=12cm]{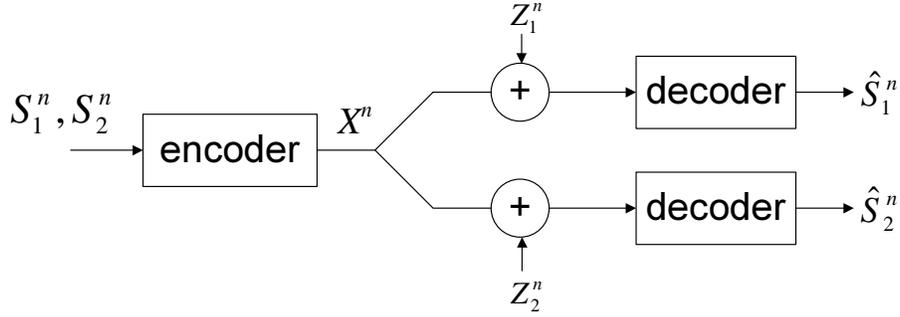}
\caption{Broadcasting bivariate Gaussian sources.\label{fig:systemdiag}}
\end{centering}
\end{figure}

The mean squared error distortion measure is used, which is given by  
$d(s_k^n,\hat{s}_k^n)=\frac{1}{n}\sum_{i=1}^n(s_k(i)-\hat{s}_k(i))^2$ for $k=1,2$. The encoder maps a source sample block $(S^n_1,S^n_2)$ into a channel input block $X^n$; the decoder observing channel output block $Y^n_k$ reconstruct the source $S^n_k$ within certain distortion; see Fig. \ref{fig:systemdiag}. The channel input $X$ is subject to an average power constraint. More formally, the problem is defined as follows.
\begin{definition}
An $(n,P,d_1,d_2)$ bivariate Gaussian source-channel broadcast code is given by an encoding function
\begin{align}
f:\mathbb{R}^n\times \mathbb{R}^n\rightarrow \mathbb{R}^n,
\end{align}
such that
\begin{align}
\frac{1}{n}\sum_{i=1}^n\Expt (X(i))^2\leq P,
\end{align}
and two decoding functions
\begin{align}
g_k:\mathbb{R}^n\rightarrow \mathbb{R}^n,\quad k =1,2,
\end{align}
and their induced distortions
\begin{align}
d_k=\Expt d(S^n_k,g_k(f(S_1^n,S_2^n)+Z^n_k)),\quad k=1,2.
\end{align}
where $\Expt(\cdot)$ is the expectation operation.
\end{definition}

In the definition, $+$ in the expression $f(S_1^n,S_2^n)+Z^n_k$ is understood as the length-$n$ vector addition. It is clear that the performance of any Gaussian joint source-channel code depends only on the marginal distribution of $(S_1^n,S_2^n,X^n,Y^n_k)$, but not the joint distribution $(S_1^n,S_2^n,X^n,Y^n_1,Y^n_2)$. Note that this implies that physical degradedness does not differ from statistical degradedness in terms of the system performance. Since the Gaussian broadcast channel is always statistically degraded, we can assume physical degradedness without loss of generality.

\begin{definition}
A distortion pair $(D_1,D_2)\in \mathbb{R}_+\times\mathbb{R}_+$ is achievable under power constraint $P$, if for any $\epsilon>0$ and sufficiently large $n$ there exists an $(n,P,d_1,d_2)$ bivariate Gaussian source-channel broadcast code such that
\begin{align}
D_i+\epsilon\geq d_i,\quad i=1,2.
\end{align}
\end{definition}

The collection of all the achievable distortion pairs under power constraint $P$ for a given bivariate source is denoted by $\mathcal{D}(P,\sigma^2,\rho,N_1,N_2)$, and this is the region we shall characterize in this work. In fact, we shall determine the following function which clearly provides a characterization of the achievable distortion region.
\begin{align}
D_2(P,\sigma^2,\rho,N_1,N_2,D_1)=\min_{(D_1,d_2)\in \mathcal{D}(P,\sigma^2,\rho,N_1,N_2)}d_2.
\end{align}
Note that $\mathcal{D}(P,\sigma^2,\rho,N_1,N_2)$ is a closed set, and thus the minimization above is meaningful. Since the minimum $D_1$ that is achievable is given by
\begin{align}
D^{\min}_1\triangleq \frac{N_1\sigma^2}{P+N_1},
\end{align}
when the second receiver is completely ignored, the function $D_2(P,\sigma^2,\rho,N_1,N_2,D_1)$ is thus only meaningfully defined on the domain $[D^{\min}_1,\infty]$.

When the source is not Gaussian but with the same covariance structure, and the channel noises are not Gaussian, but has the same variances, we shall denote the achievable distortion region as $\mathcal{D}^*(P,\sigma^2,\rho,N_1,N_2)$. We shall show that the Gaussian setting has the worst case property, and thus the Gaussian scheme is a ``robust'' scheme.

\section{Main Result}
\label{sec:mainresult}

The main result of this paper is summarized in the following theorem.

\begin{theorem}
\label{theorem:maintheorem}
If $D_1> D^{\max}_1$, where $D^{\max}_1 \triangleq\sigma^2\frac{(1-\rho^2)P+N_1}{P+N_1}$, then
\begin{align}
D_2(P,\sigma^2,\rho,N_1,N_2,D_1)=\sigma^2\frac{N_2}{P+N_2}.
\end{align}

If $P\leq \frac{2\rho N_1}{1-\rho}$, then in the range $D_1\in [D^{\min}_1,D^{\max}_1]$, 
\begin{align}
D_2(P,\sigma^2,\rho,N_1,N_2,D_1)=D_2^u(P,\sigma^2,\rho,N_1,N_2,D_1),
\end{align}
where
\begin{align}
&D_2^u(P,\sigma^2,\rho,N_1,N_2,D_1)\nonumber\\
&\qquad\triangleq\sigma^2\Bigg{[}\bigg{(}\sqrt{1-\frac{D_1(P+N_1)}{P\sigma^2}+\frac{N_1}{P}}-\sqrt{\frac{\rho^2}{1-\rho^2}(\frac{(P+N_1)D_1}{P\sigma^2}-\frac{N_1}{P})}\bigg{)}^2\frac{(1-\rho^2)P}{P+N_2}+\frac{N_2}{P+N_2}\Bigg{]}.
\end{align}
On the other hand if $P> \frac{2\rho N_1}{1-\rho}$, then 
\begin{align}
D_2(P,\sigma^2,\rho,N_1,N_2,D_1)=\Bigg{\{}\begin{array}{cc}
D_2^u(P,\sigma^2,\rho,N_1,N_2,D_1)& D_1\in [D^{\min}_1,D^-_1)\cup (D^+_1,D^{\max}_1]\\
D_2^h(P,\sigma^2,\rho,N_1,N_2,D_1)&D_1\in [D^-_1,D^+_1]
\end{array}
\end{align}
where
\begin{align}
D_2^h(P,\sigma^2,\rho,N_1,N_2,D_1)\triangleq\frac{\sigma^2}{P+N_2}\Bigg{[}\frac{N_1(1-\rho^2)\sigma^2}{D_1}+N_2-N_1\Bigg{]}
\end{align}
and 
\begin{align}
&D^-_1=\sigma^2\frac{(P+2N_1)(1-\rho^2)-\sqrt{(P^2-(P+2N_1)^2\rho^2)(1-\rho^2)}}{2(P+N_1)}\\
&D^+_1=\sigma^2\frac{(P+2N_1)(1-\rho^2)+\sqrt{(P^2-(P+2N_1)^2\rho^2)(1-\rho^2)}}{2(P+N_1)}.
\end{align}
\end{theorem}

\vspace{0.5cm}
\noindent\textit{Remark:} Depending on the power constraint, the achievable distortion region may have two operating regimes. In the regime where $D_2(P,\sigma^2,\rho,N_1,N_2,D_1)=D^u_2(P,\sigma^2,\rho,N_1,N_2,D_1)$, the uncoded scheme given in \cite{Bross:08} is optimal, whereas in the regime $D_2(P,\sigma^2,\rho,N_1,N_2,D_1)=D^h_2(P,\sigma^2,\rho,N_1,N_2,D_1)$, a hybrid scheme given in the next section is optimal, but the uncoded scheme is not. Typical achievable distortion regions for these two cases will be illustrated in the next section after the hybrid coding scheme is given, where more observations regarding these schemes can be discussed.

Consider a source pair $(S^*_1,S^*_2)$ whose covariance is given by (\ref{eqn:covariance}), and channel noise pair $(Z^*_1,Z^*_2)$ whose variances are given by $(N_1,N_2)$. We have the following theorem.
\begin{theorem}\label{theorem:worstcase}
If $(D_1,D_2)\in \mathcal{D}(P,\sigma^2,\rho,N_1,N_2)$, then $(D_1,D_2)\in \mathcal{D}^*(P,\sigma^2,\rho,N_1,N_2)$. 
\end{theorem}

This theorem essentially says the Gaussian setting has the worst sources and channels among those having the same covariance structure, a result similar to the well-known ones that the Gaussian source is the worst source \cite{CoverThomas} (Ex. 13.8) and the Gaussian channel is the worst channel \cite{CoverThomas} (Ex. 10.1).   

The proofs of Theorem \ref{theorem:maintheorem} and Theorem \ref{theorem:worstcase} are given in the next section. 

\section{Proof of the Main Result}
\label{sec:proof}

In this section, we shall first review some previous results which provide partial characterization of the achievable distortion region, then give a new hybrid coding scheme, which  provides the missing portion of the characterization. Finally, we provide a proof for the worst-case property of the Gaussian setting. 

\subsection{Preliminaries and Previous Results}
\label{subsec:previous}

It is straightforward to show that a simple analog scheme by sending $S_2$ directly (after certain scaling) achieves the distortion pair (see also \cite{Bross:08})
\begin{align}
D_1=\sigma^2\frac{(1-\rho^2)P+N_1}{P+N_1}= D^{\max}_1,\qquad \mbox{and} \qquad D_2=\sigma^2\frac{N_2}{P+N_2},
\end{align}
for which $D_2$ cannot be reduced even when the first receiver is not present. Thus we trivially have 
\begin{align}
D_2(P,\sigma^2,\rho,N_1,N_2,D_1)=\sigma^2\frac{N_2}{P+N_2},\qquad \mbox{if}\qquad D_1\geq \sigma^2\frac{(1-\rho^2)P+N_1}{P+N_1}.
\end{align}
Thus we only need to characterize the function $D_2(P,\sigma^2,\rho,N_1,N_2,D_1)$ when $D_1\in[D^{\min}_1,D^{\max}_1]$.

The uncoded scheme was investigated thoroughly in \cite{Bross:08}, which was shown to be optimal under certain conditions. More precisely, the uncoded scheme uses the single letter mapping
\begin{align}
X(i)=\sqrt{\frac{P}{\sigma^2(\alpha^2+2\alpha\beta\rho+\beta^2)}}\bigg{(}\alpha S_1(i)+\beta S_2(i)\bigg{)},
\end{align}
where $\alpha\in [0,1]$ and $\beta=1-\alpha$. The reconstruction in this uncoded scheme is thus naturally the single letter mapping given by $\Expt [S_k(i)|Y_k(i)]$ for $k=1,2$. The distortion pair is thus given by
\begin{align}
&\hat{D}_1(\alpha,\beta)=\sigma^2\left[\frac{P\beta^2(1-\rho^2)}{(P+N_1)(\alpha^2+2\alpha\beta\rho+\beta^2)}+\frac{N_1}{P+N_1}\right]\\
&\hat{D}_2(\alpha,\beta)=\sigma^2\left[\frac{P\alpha^2(1-\rho^2)}{(P+N_2)(\alpha^2+2\alpha\beta\rho+\beta^2)}+\frac{N_2}{P+N_2}\right].
\end{align}

The main result of \cite{Bross:08} is the following theorem. 

\begin{theorem}[Theorem 1, \cite{Bross:08}]
\label{theorem:Bross}
For any $(D_1,D_2)\in \mathcal{D}(P,\sigma^2,\rho,N_1,N_2)$, and
\begin{align}
\label{eqn:SNRbound}
\frac{P}{N_1}\leq \Gamma(D_1,\sigma^2,\rho),
\end{align}
there exist $\alpha^*,\beta^*\geq 0$ such that 
\begin{align}
\hat{D}_1(\alpha^*,\beta^*)\leq D_1,\qquad \mbox{and}\qquad \hat{D}_2(\alpha^*,\beta^*)\leq D_2,
\end{align}
where the threshold $\Gamma$ is given by
\begin{align}
\Gamma(D_1,\sigma^2,\rho)=\Bigg{\{}
\begin{array}{cc}
\frac{\sigma^4(1-\rho^2)-2D_1\sigma^2(1-\rho^2)+D^2_1}{D_1(\sigma^2(1-\rho^2)-D_1)}& 0<D_1<\sigma^2(1-\rho^2),\\
+\infty&\mbox{otherwise}.
\end{array}
\end{align}
\end{theorem}

This result partially characterizes the achievable distortion region $\mathcal{D}(P,\sigma^2,\rho,N_1,N_2)$, or equivalently the function $D_2(P,\sigma^2,\rho,N_1,N_2,D_1)$. Later we shall translate it into a more convenient form, such that by combining all the portions, the function $D_2(P,\sigma^2,\rho,N_1,N_2,D_1)$ can be characterized completely.

Also inherent in the proof of \cite{Bross:08} is a genie-aided outer bound, where the source $S_2$ is given directly to the decoder observing $Y_1$; see Fig. \ref{fig:genie}. This outer bound, which is a special case of the more general outer bound given in \cite{Bross:08}, was rederived in a simpler manner in \cite{Vishwanath:09}; the general outer bound in \cite{Bross:08} is more involved, and it usually requires optimization over several parameters. Let us follow the notation in \cite{Vishwanath:09} and denote the achievable distortion $(D_1,D_2)$ in the genie-aided setting as $(D_{1|2},D_2)$, and the simplified characterization of the achievable distortion region in \cite{Vishwanath:09} is as follows.
\begin{theorem}\label{theorem:genie}
The achievable $(D_{1|2},D_2)$ when $S_2$ is present at the first receiver is all the pairs 
\begin{align}
D_{1|2}\geq \frac{\sigma^2(1-\rho^2)}{1+\frac{\alpha P}{N_1}},\qquad \mbox{and}\qquad  D_2\geq \frac{\sigma^2}{1+\frac{(1-\alpha) P}{\alpha P+N_2}},
\end{align}
for some $\alpha\in [0,1]$.
\end{theorem}

As we shall show in the next subsection, this genie-aided outer bound is in fact tight for a certain regime. Intuitively speaking, since the first receiver is stronger than the second receiver, and $S_2$ is not required at the first receiver, the information $S_2$ is redundant at the first receiver in a certain sense, and thus we can expect the genie-aided outer bound to be reasonably good.

\begin{figure}[tb]
\begin{centering}
\includegraphics[width=12cm]{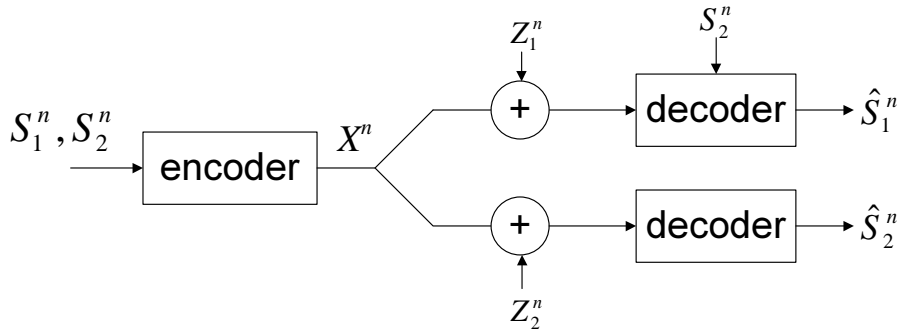}
\caption{Broadcasting bivariate Gaussian sources with a genie helper.\label{fig:genie}}
\end{centering}
\end{figure}

\subsection{A Hybrid Coding Scheme}

The coding scheme we propose is a hybrid one, where the channel input is given by 
\begin{align}
X^n=\tilde{\alpha} S^n_1+\tilde{\beta} S^n_2+ X_d^n,
\end{align}
where $X_d^n$ is (roughly) the quantized version of source sequence $S^n_2$ after some proper scaling;  $X^n_d$ is the digital portion of the channel input, and $\tilde{\alpha},\tilde{\beta}\geq 0$ are two scaling parameters to be specified later. 

More precisely, consider the single-letter distribution
\begin{align}
X_d=\tilde{\gamma}(S_2+U),
\end{align}
where $U$ is a zero mean Gaussian random variable independent of everything else with variance $Q$, and $\tilde{\gamma}\geq 0$ is again a scaling parameter to be specified later. 
Since $X_d\leftrightarrow S_2\leftrightarrow (S_1,S_2)$ is a Markov string, the joint distribution of $(S_1,S_2,X_d,X,Y_1,Y_2)$ is uniquely determined, and they are jointly Gaussian. We also need to define the coefficients $a_k$ and $b_k$ 
\begin{align}
\Expt[S_k|X_d,\tilde{\alpha} S_1+\tilde{\beta} S_2+Z_k]=a_kX_d+b_k(\tilde{\alpha} S_1+\tilde{\beta} S_2+Z_k), \qquad k=1,2.
\end{align}

This proposed hybrid scheme in this work is somewhat similar to the scheme given in \cite{Lapidoth:09} for joint source channel coding on the multiple access channel. In what follows, we only outline the coding scheme and some important analysis steps, but omit the rather technical detailed proof (a rigorous proof can be straightforwardly adapted from that given in \cite{Lapidoth:09}). 
\begin{itemize}
\item Codebook generation: generate $\exp(nR)$ codewords single-letter wise according to the marginal distribution of $X_d$; this codebook is revealed to both the encoder and decoders. 
\item Encoding: find a sequence $X_d^n$ in the codebook that is jointly typical with the source sequence $S^n_2$; if successful, the transmitter sends $X^n=\tilde{\alpha} S^n_1+\tilde{\beta} S^n_2+ X_d^n$.
\item Digital decoding: the $k$-th decoder tries to find a unique $X^n_d$ codeword in the codebook that is jointly typical with $Y^n_k$; the decoder also recovers the sequence $\tilde{\alpha} S^n_1+\tilde{\beta} S^n_2+Z^n_k$ after removing the digital codeword.
\item Estimation: if the digital decoding succeeds, then the decoder reconstructs the respective source sequence as $\hat{S}_k(i)=a_kX_d(i)+b_k(\tilde{\alpha} S_1(i)+\tilde{\beta} S_2(i)+Z_k(i))$.
\end{itemize}

An error occurs in the above scheme if the encoder fails to find a codeword that is jointly typical with $S^n_2$, or one of the decoders fails to find the correct digital codeword. Note that due to the Markov string $X_d\leftrightarrow S_2\leftrightarrow (S_1,S_2)$, we indeed have that the chosen $X_d^n$ is jointly typical with $(S^n_1,S^n_2,X^n,Y^n_1,Y^n_2)$ with high probability in the above scheme. 
Because the second receiver is a degraded version of the first receiver, the error probability can be made arbitrarily small if the following condition holds (after ignoring the $\delta$'s often seen in the typicality argument)
\begin{align}
\label{eqn:rateconstraint}
I(S_2;X_d)\leq R \leq I(X_d;Y_2).
\end{align}
Furthermore, to ensure the power constraint is not violated, we need 
\begin{align}
\label{eqn:powerconstraint}
\sigma^2(\tilde{\alpha}^2+\tilde{\beta}^2+2\rho\tilde{\alpha}\tilde{\beta})+\tilde{\gamma}^2(\sigma^2+Q)+2\tilde{\gamma}\sigma^2(\tilde{\alpha}\rho+\tilde{\beta})\leq P.
\end{align}

It is evident that as long as
\begin{align}
\label{eqn:reducedpowercondition}
\sigma^2(\tilde{\alpha}^2+\tilde{\beta}^2+2\rho\tilde{\alpha}\tilde{\beta})\leq P,
\end{align}
we can find $\tilde{\gamma}=\tilde{\gamma}^*$ such that (\ref{eqn:powerconstraint}) holds with equality, because the left hand side of (\ref{eqn:powerconstraint}) is monotonically increasing in $\tilde{\gamma}$ in the range $[0,\infty]$, and we shall choose precisely this value in the scheme. 

With this choice of $\tilde{\gamma}$, (\ref{eqn:rateconstraint}) can be simplified. Note that
\begin{align}
I(X_d;Y_2)&=h(Y_2)-h(Y_2|X_d)\nonumber\\
&=h(Y_2)-h(\tilde{\alpha} S_1+\tilde{\beta} S_2+ X_d+Z_2|X_d)\\
&=h(Y_2)-h(\tilde{\alpha} S_1+\tilde{\beta} S_2+Z_2|X_d)\\
&=\frac{1}{2}\log 2\pi e(P+N_2)-h(\tilde{\alpha} S_1+\tilde{\beta} S_2+Z_2|X_d)\\
&=\frac{1}{2}\log 2\pi e(P+N_2)-h(\tilde{\alpha} S_1+\tilde{\beta} S_2+Z_2|S_2+U).
\label{eqn:leftside}
\end{align}
Moreover, we have also
\begin{align}
I(S_2;X_d)=I(S_2;S_2+U)=\frac{1}{2}\log\frac{\sigma^2+Q}{Q}.\label{eqn:rightside}
\end{align}
Putting (\ref{eqn:leftside}) and (\ref{eqn:rightside}) together, (\ref{eqn:rateconstraint}) finally reduces to the following expression through some algebraic calculation
\begin{align}
Q\geq \sigma^2\frac{(1-\rho^2)\tilde{\alpha}^2\sigma^2+N_2}{P-(\tilde{\alpha}^2+\tilde{\beta}^2+2\tilde{\alpha}\tilde{\beta}\rho)\sigma^2}\triangleq Q^*.
\end{align}
We shall thus choose $Q=Q^*$ in the scheme (again ignoring some asymptotically small $\delta$ terms). 

By the joint typicality of the sequences when digital decoding succeeds, the distortions can be computed, which is summarized in the following theorem. 

\begin{theorem}
For $\tilde{\alpha},\tilde{\beta}\geq 0$ such that
\begin{align}
\sigma^2(\tilde{\alpha}^2+\tilde{\beta}^2+2\rho\tilde{\alpha}\tilde{\beta})\leq P,
\end{align}
the proposed hybrid scheme can achieve the distortion pair $(D^h_1,D^h_2)$
\begin{align}
&\tilde{D}^{h}_1=\sigma^2\frac{(1-\rho^2)[N_1P-(\tilde{\beta}^2+2\tilde{\alpha}\tilde{\beta}\rho)N_1\sigma^2+\tilde{\beta}^2N_2\sigma^2+(1-\rho^2)\tilde{\alpha}^2\tilde{\beta}^2\sigma^4]+N_1N_2}{(1-\rho^2)\tilde{\alpha}^2(P+N_1)\sigma^2+PN_1+N_1N_2+(\tilde{\alpha}^2+\tilde{\beta}^2+2\tilde{\alpha}\tilde{\beta}\rho)(N_2-N_1)\sigma^2},\nonumber\\
&\tilde{D}^h_2=\sigma^2\frac{\tilde{\alpha}^2(1-\rho^2)\sigma^2+N_2}{P+N_2}.\label{eqn:innerbound}
\end{align}
\end{theorem}

It is not clear a priori why in the analog part, the source $S_2$ needs to be included. However, this is rather critical, as we shall discuss shortly in the next section. 

\subsection{Matching the Inner and Outer Bounds}

The proof of the main result is now organized into three propositions. We start by rewriting Theorem \ref{theorem:Bross} in the following form. 
\begin{prop}
We have 
\begin{itemize}
\item If $P\leq \frac{2\rho N_1}{1-\rho}$, then the uncoded scheme is optimal in the range $D_1\in [D^{\min}_1,D^{\max}_1]$ and $D_2(P,\sigma^2,\rho,N_1,N_2,D_1)=D^u_2(P,\sigma^2,\rho,N_1,N_2,D_1)$;
\item If $P> \frac{2\rho N_1}{1-\rho}$, then the uncoded scheme is optimal over the range $D_1\in [D^{\min}_1,D^-_1]\cup [D^+_1,D^{\max}_1]$, and in this range $D_2(P,\sigma^2,\rho,N_1,N_2,D_1)=D^u_2(P,\sigma^2,\rho,N_1,N_2,D_1)$.
\end{itemize}
\end{prop}

\begin{proof}
This proposition is a direct consequence of Theorem \ref{theorem:Bross}. It is straightforward to show that when $P\leq \frac{2\rho N_1}{1-\rho}$, the condition in (\ref{eqn:SNRbound}) holds for any choice of $D_1$ by analyzing the quadratic inequality (see \cite{Bross:08} Corollary 1); on the other hand, when $P> \frac{2\rho N_1}{1-\rho}$, (\ref{eqn:SNRbound}) only holds in the given range. 
\end{proof}

Given this alternative form, we only need to focus on the case that $P> \frac{2\rho N_1}{1-\rho}$, and moreover, in the range $D_1\in [D^-_1,D^+_1]$. It is worth noting that it is always true that $D^{\min}_1\leq D^-_1\leq D^+_1\leq (1-\rho^2)\sigma^2\leq D^{\max}_1$. 

Next we use Theorem \ref{theorem:genie} to write an lower bound for the function $D_2(P,\sigma^2,\rho,N_1,N_2,D_1)$.
\begin{prop}
\label{prop:lowerbound}
For any $D_1\in[D^{\min}_1,D^{\max}_1]$, $D_2(P,\sigma^2,\rho,N_1,N_2,D_1)\geq D^h_2(P,\sigma^2,\rho,N_1,N_2,D_1)$.
\end{prop}

This proposition is a direct consequence of the fact that the achievable $(D_{1|2},D_2)$ region given in Theorem \ref{theorem:genie} is an outer bound for $\mathcal{D}(P,\sigma^2,\rho,N_1,N_2)$; the proof is obtained by simple algebraic manipulation to eliminate the parameter $\alpha$ in Theorem \ref{theorem:genie}, and we thus omitted the details.

The following proposition is the final piece for the proof of the main result. 
\begin{prop}
The proposed hybrid scheme achieves 
\begin{align}
D_2(P,\sigma^2,\rho,N_1,N_2,D_1)=D^h_2(P,\sigma^2,\rho,N_1,N_2,D_1),
\label{eqn:matchingbound}
\end{align}
when $D_1\in [D^-_1,D^+_1]$ by choosing 
\begin{align}
&\tilde{\alpha}=\sqrt{\frac{N_1}{D_1}-\frac{N_1}{\sigma^2(1-\rho^2)}}\label{eqn:alphachoice}\\
&\tilde{\beta}=\frac{N_1\rho}{\tilde{\alpha} (1-\rho^2)\sigma^2}.
\label{eqn:betachoice}
\end{align}
\end{prop}
\begin{proof}
In order to prove this proposition, we need to show firstly that the given choice of $(\tilde{\alpha},\tilde{\beta})$ does not violate the power constraint, i.e., the condition (\ref{eqn:reducedpowercondition}) is satisfied; secondly, the given choice of $(\tilde{\alpha},\tilde{\beta})$ reduces the distortion pairs in (\ref{eqn:innerbound}) to those given in (\ref{eqn:matchingbound}). 

Notice that $\tilde{\alpha}$ and $\tilde{\beta}$ given in (\ref{eqn:alphachoice}) and (\ref{eqn:betachoice}) are well-defined and non-negative when $D_1\in [D^{-}_1,D^{+}_1]$, since $D_1\leq (1-\rho^2)\sigma^2$ in this range. For (\ref{eqn:reducedpowercondition}) to hold, we need to have
\begin{align}
\tilde{\alpha}^2+\tilde{\beta}^2+2\tilde{\alpha}\tilde{\beta}\rho=\tilde{\alpha}^2+\frac{\rho^2N^2_1}{\tilde{\alpha}^2(1-\rho^2)^2\sigma^4}+\frac{2\rho^2 N_1}{ (1-\rho^2)\sigma^2}\leq \frac{P}{\sigma^2}.
\end{align}
Solving the inequality gives that the necessary and sufficient condition that
\begin{align}
&\frac{(1-\rho^2)P-2\rho^2N_1-\sqrt{(P^2-(P+2N_1)^2\rho^2)(1-\rho^2)}}{2(1-\rho^2)\sigma^2}\nonumber\\
&\qquad\leq \tilde{\alpha}^2\leq \frac{(1-\rho^2)P-2\rho^2N_1+\sqrt{(P^2-(P+2N_1)^2\rho^2)(1-\rho^2)}}{2(1-\rho^2)\sigma^2}.
\label{eqn:middlealpha}
\end{align}
Substituting (\ref{eqn:alphachoice}) into (\ref{eqn:middlealpha}) and after certain algebra, we arrive at the conditions under which (\ref{eqn:middlealpha}) is true, which is exactly $D_1\in [D^-_1,D^+_1]$. 

In order to show that the lower bound as stated in Proposition \ref{prop:lowerbound} can be achieved, we first simplify the expression of $\tilde{D}^h_1$ given in (\ref{eqn:innerbound}) in terms of $\tilde{\alpha}$, which (after quite some algebra) eventually gives
\begin{align}
\tilde{D}^h_1=\sigma^2\frac{(1-\rho^2)N_1}{(1-\rho^2)\tilde{\alpha}^2\sigma^2+N_1}.
\label{eqn:almostfinalD1}
\end{align}
Substituting our choice of $\tilde{\alpha}$ given in (\ref{eqn:alphachoice}) into (\ref{eqn:almostfinalD1}) leads to $\tilde{D}^h_1=D_1$; again substituting (\ref{eqn:alphachoice}) into the expression of $\tilde{D}^h_2$ in (\ref{eqn:innerbound}) gives the expression stated in the proposition, which completes the proof.
\end{proof}

\vspace{0.5cm}
Readers may wonder how the magic value of $\tilde{\beta}$ was found, which optimizes $\tilde{D}^h_1$ in the hybrid scheme. Indeed, directly optimizing the distortion $\tilde{D}^h_1$ is extremely cumbersome\footnote{In fact we were not able to find the optimal solution for $\tilde{\beta}$ this way.}. To circumvent this difficulty, we instead solve for $\tilde{\beta}$ such that the inner bound matches the outer bound, which gives the given expression. This approach is less intuitive, since it is possible that neither the genie-aided outer bound nor the hybrid scheme inner bound is tight, however extensive numerical comparison indeed suggested that these bounds  match, which motivated us to take such an approach.

In Fig. \ref{fig:regime1} and Fig. \ref{fig:regime2} we give two typical achievable distortion regions, where for comparison we also include the performance of a simple separation-based scheme where the digital broadcast messages encoding $S_2$ and $S_1$ (conditioned on the reconstructed $S_2$), respectively. 
In both figures, each horizontal red line is the performance of the hybrid scheme by varying $\tilde{\beta}$ while keeping $\tilde{\alpha}$ fixed. Note that the hybrid scheme includes the uncoded scheme as a special case when the digital portion is allocated no power. Fig. \ref{fig:regime1} is plotted with the choice of source and channel satisfying the condition $P\leq \frac{2\rho N_1}{1-\rho}$, and thus the uncoded scheme is always optimal. For this case, adding digital code in the hybrid scheme is always inferior. In contrast, Fig \ref{fig:regime2} is plotted under the condition $P> \frac{2\rho N_1}{1-\rho}$, and thus uncoded scheme is only optimal at high and low $D_1$ regimes. In the regime that uncoded scheme is not optimal, it can be seen that even when analog portion does not include $S_2$, the hybrid scheme can sometimes outperform the uncoded scheme, however, by optimizing $S_2$ in the analog portion, the inner and outer bounds indeed match. Moreover, observe that the distortion of $\tilde{D}^h_1$ achieved by the hybrid scheme is not monotonic in $\tilde{\beta}$ when $\tilde{\alpha}$ is fixed (each red line), where the two extreme values of $\tilde{\beta}$ give the uncoded scheme and the hybrid scheme without analog $S_2$, respectively. 

\begin{figure}[tb]
\begin{centering}
\includegraphics[width=17cm]{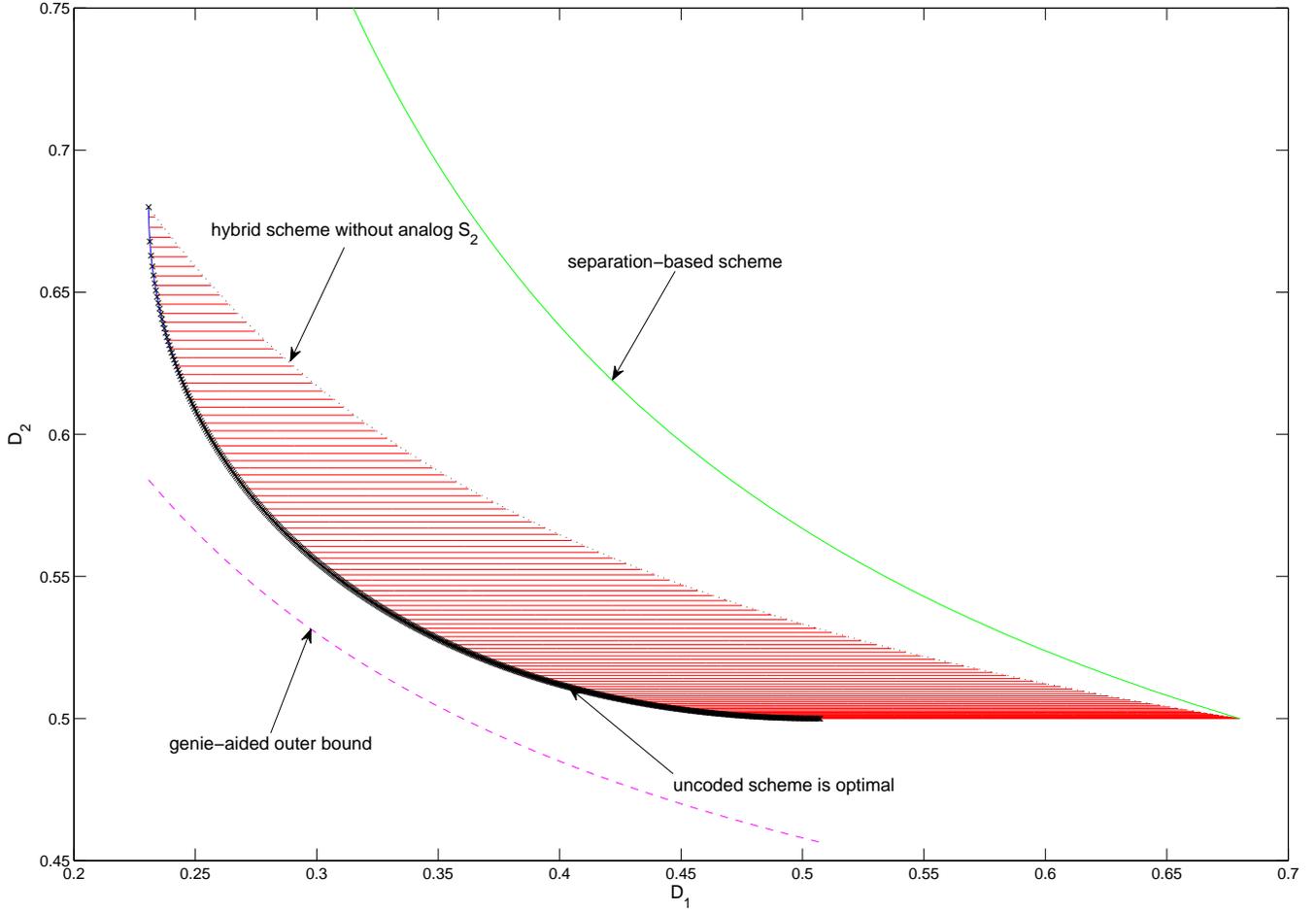}
\caption{A typical achievable distortion region when $P\leq \frac{2\rho N_1}{1-\rho}$. Here $P=1$, $N_1=0.3$, $N_2=1$ and $\rho=0.8$.\label{fig:regime1}}
\end{centering}
\end{figure}

\begin{figure}[tb]
\begin{centering}
\includegraphics[width=17cm]{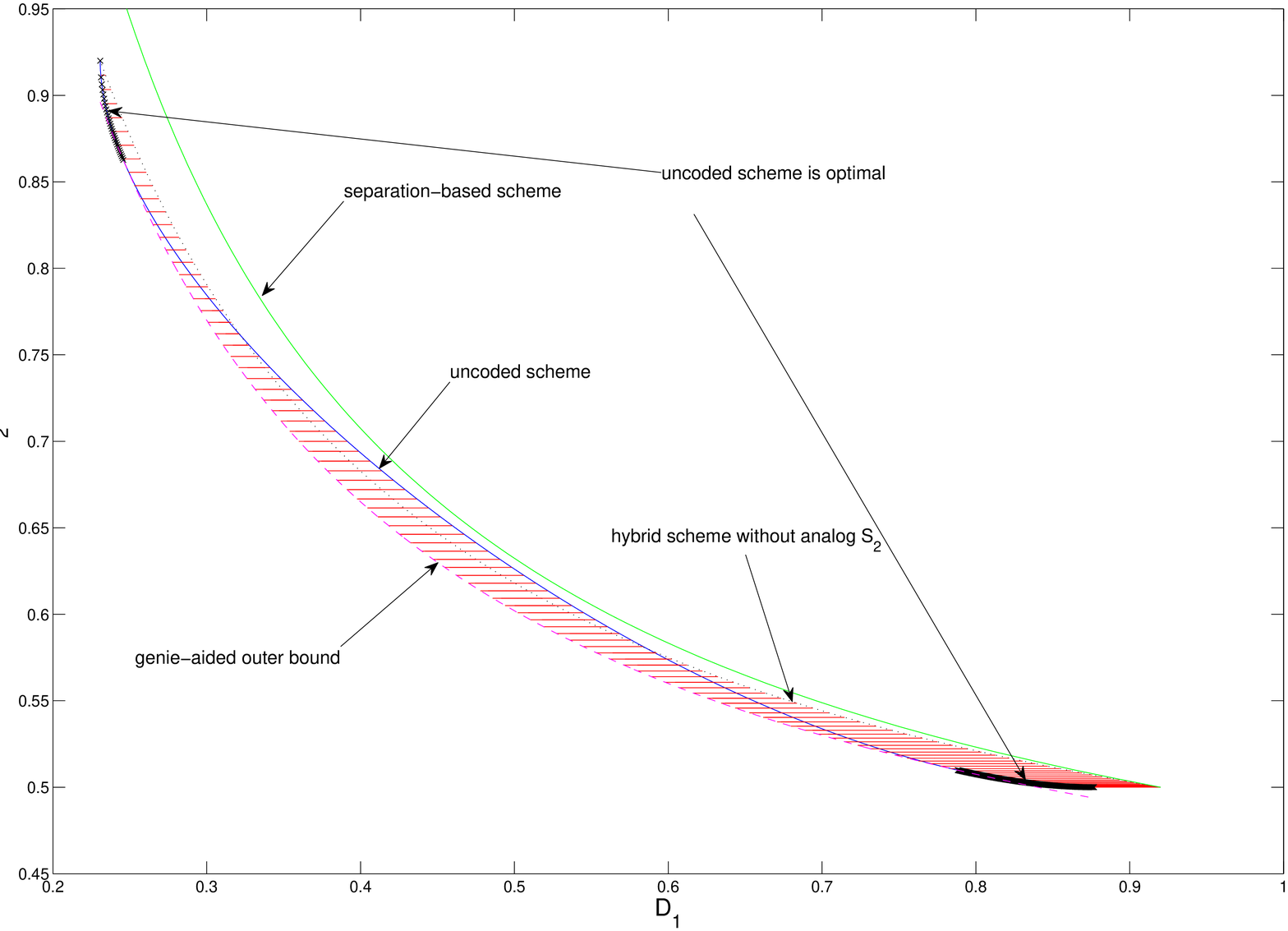}
\caption{A typical achievable distortion region when $P> \frac{2\rho N_1}{1-\rho}$. Here $P=1$, $N_1=0.3$, $N_2=1$ and $\rho=0.4$.\label{fig:regime2}}
\end{centering}
\end{figure}

\subsection{The Worst Case Property}

Next we prove Theorem \ref{theorem:worstcase}, i.e., the worst case property of the Gaussian setting. 

\begin{proof}[Proof of Theorem \ref{theorem:worstcase}]
We have shown an optimal scheme in the Gaussian setting is the proposed hybrid scheme, and thus we can limit ourselves to the distortion pairs achievable by this scheme. In fact we shall continue to use this scheme and the associated parameters when the sources and channel noises are not Gaussian. More precisely, we shall now use $X^*_d$ instead of $X_d$ to construct the digital source codewords
\begin{align}
X^*_d=\tilde{\gamma}(S^*_2+U),
\end{align}
where $U$ is still a Gaussian random variable with variance of $Q^*$, independent of everything else. The overall covariance structure of the scheme remains intact as in the Gaussian case, and thus the same (MSE) distortion pairs can be achieved, as long as the digital codewords can be correctly decoded at both the decoders, i.e.,
\begin{align}
I(S^*_2;X^*_d)\leq I(X^*_d;Y^*_i), \qquad i=1,2,
\end{align}
with our choices of the parameters, where $Y^*_i$ is the channel output in this non-Gaussian setting. Note that unlike in the Gaussian case, here the broadcast channel is not necessarily degraded, and thus we also need to make sure that the codeword can be correctly decoded at the first decoder.

To show the second decoder can succeed (with high probability), we only need to observe that
\begin{align}
I(S^*_2;X^*_d)-I(X^*_d;Y^*_2)&=h(X^*_d|Y^*_2)-h(X^*_d|S^*_2)\nonumber\\
&=h(X^*_d|Y^*_2)-h(\tilde{\gamma}U)\nonumber\\
&\leq h(X^*_d-\Expt(X^*_d|Y^*_2))-h(\tilde{\gamma}U)\nonumber\\
&\leq h(X_d-\Expt(X_d|Y_2))-h(\tilde{\gamma}U)\nonumber\\
&=h(X_d|Y_2)-h(\tilde{\gamma}U)\nonumber\\
&=I(S_2;X_d)-I(X_d;Y_2)=0,
\end{align}
where in the second inequality we substitute $X_d$ of the Gaussian version of the problem, because the terms have the same covariance structure, and Gaussian distribution maximizes the differential entropy; in the last but one equality, we add and subtract the same term $h(X_d)$, and the last equality is due to our specific choice of the parameters in the Gaussian problem. 

Similarly, we can write
\begin{align}
I(S^*_2;X^*_d)-I(X^*_d;Y^*_1)\leq I(S_2;X_d)-I(X_d;Y_1)\leq I(S_2;X_d)-I(X_d;Y_2)\leq 0,
\end{align}
where the second inequality is guaranteed by the relation in the Gaussian case, which is indeed a  degraded broadcast channel. This completes the proof.
\end{proof} 

\section{Concluding Remarks}
\label{sec:conclusion}

We provide a complete solution for the joint source-channel coding problem of sending bivariate Gaussian sources over Gaussian broadcast channels when the source bandwidth and channel bandwidth are matched. Thus this problem joins a limited list of joint source-channel coding problems for which complete solutions are known. Possible extension of this work  includes the case with more than two users or more than two sources, and approximate characterization for bandwidth mismatched case, which are part of our on-going work.

\bibliographystyle{IEEEbib}

\end{document}